\documentclass[a4paper]{article}
\usepackage[utf8]{inputenc}

\usepackage[margin=1.3in]{geometry}

\usepackage{amsthm,amssymb,amsfonts,amsmath}
\usepackage{amssymb,amsfonts,amsmath}
\usepackage{xcolor}
\newtheorem{theorem}{Theorem}
\newtheorem{lemma}[theorem]{Lemma}
\newtheorem{corollary}[theorem]{Corollary}

\newtheorem{conjecture}{Conjecture}

\usepackage{graphicx}
\graphicspath{{figures/}}

\usepackage{thm-restate}
\usepackage{url}
\usepackage{faktor}
\usepackage{tikz-cd}

\title{Ham-Sandwich cuts and center transversals in subspaces}

\author{Patrick Schnider\thanks{Department of Computer Science,
        ETH Z\"{u}rich, Switzerland. {\tt patrick.schnider@inf.ethz.ch}}}
        
\date{}

\begin{document}

\maketitle

\begin{abstract}
The Ham-Sandwich theorem is a well-known result in geometry.
It states that any $d$ mass distributions in $\mathbb{R}^d$ can be simultaneously bisected by a hyperplane.
The result is tight, that is, there are examples of $d+1$ mass distributions that cannot be simultaneously bisected by a single hyperplane.
In this abstract we will study the following question: given a continuous assignment of mass distributions to certain subsets of $\mathbb{R}^d$, is there a subset on which we can bisect more masses than what is guaranteed by the Ham-Sandwich theorem?

We investigate two types of subsets.
The first type are linear subspaces of $\mathbb{R}^d$, i.e., $k$-dimensional flats containing the origin.
We show that for any continuous assignment of $d$ mass distributions to the $k$-dimensional linear subspaces of $\mathbb{R}^d$, there is always a subspace on which we can simultaneously bisect the images of all $d$ assignments.
We extend this result to center transversals, a generalization of Ham-Sandwich cuts.
As for Ham-Sandwich cuts, we further show that for $d-k+2$ masses, we can choose $k-1$ of the vectors defining the $k$-dimensional subspace in which the solution lies.

The second type of subsets we consider are subsets that are determined by families of $n$ hyperplanes in $\mathbb{R}^d$.
Also in this case, we find a Ham-Sandwich-type result.
In an attempt to solve a conjecture by Langerman about bisections with several cuts, we show that our underlying topological result can be used to prove this conjecture in a relaxed setting.
\end{abstract}

\section{Introduction}
The famous Ham-Sandwich theorem (see e.g.\ \cite{Matousek, StoneTukey}, Chapter 21 in \cite{Handbook}) is a central result in geometry that initiated a significant amount of research on several ways to partition mass distributions.
It states that any $d$ mass distributions in $\mathbb{R}^d$ can be simultaneously bisected by a hyperplane.
A ($d$-dimensional) \emph{mass distribution} $\mu$ on $\mathbb{R}^d$ is a measure on $\mathbb{R}^d$ such that all open subsets of $\mathbb{R}^d$ are measurable, $0<\mu(\mathbb{R}^d)<\infty$ and $\mu(S)=0$ for every lower-dimensional subset $S$ of $\mathbb{R}^d$.
An intuitive example of a mass distribution is, for example, the volume of some full-dimensional geometric object in $\mathbb{R}^d$.
The Ham-Sandwich theorem has been generalized in several ways.
One famous generalization is the polynomial Ham-Sandwich theorem, which states that any $\binom{n+d}{d}-1$ mass distributions in $\mathbb{R}^d$ can be simultaneously bisected by an algebraic surface of degree $n$ \cite{StoneTukey}.
Another extension is the center transversal theorem, which generalizes the result to flats of lower dimensions:

\begin{theorem}[Center transversal theorem \cite{dolnikov, zivaljevic}]
Let $\mu_1,\ldots,\mu_k$ be $k$ mass distributions in $\mathbb{R}^d$, where $k\leq d$.
Then there is a $(k-1)$-dimensional affine subspace $g$ such that every halfspace containing $g$ contains at least a $\frac{1}{d-k+2}$-fraction of each mass.
\end{theorem}

We call such an affine subspace a \emph{$(k-1,d)$-center transversal}.
For $k=d$, we get the statement of the Ham-Sandwich theorem.
Further, for $k=1$, we get another well-known result in geometry, the so called \emph{Centerpoint theorem} \cite{rado}.

In this work we will consider two different generalizations of the Ham-Sandwich theorem.
The first one is about Ham-Sandwich cuts in linear subspaces.
More precisely, we define a \emph{mass assignment} on $G_k(\mathbb{R}^d)$ as a continuous assignment $\mu: G_k(\mathbb{R}^d)\rightarrow M_k$, where $G_k(\mathbb{R}^d)$ denotes the \emph{Grassmann manifold} consisting of all $k$-dimensional linear subspaces of $\mathbb{R}^d$ and $M_k$ denotes the space of all $k$-dimensional mass distributions.
In other words, $\mu$ continuously assigns a mass distribution $\mu^h:=\mu(h)$ to each $k$-dimensional linear subspace $h$ of $\mathbb{R}^d$.
Examples of mass assignments include projections of higher dimensional mass distributions to $h$ or the volume of intersections of $h$ with (sufficiently smooth) higher dimensional geometric objects.
Also, mass distributions in $\mathbb{R}^d$ can be viewed as mass assignments on $G_d(\mathbb{R}^d)$.
In fact, in this paper, we will use the letter $\mu$ both for mass distributions as well as for mass assignments.
The Ham-Sandwich theorem says that on every subspace we can simultaneously bisect the images of $k$ mass assignments.
But as there are many degrees of freedom in choosing subspaces, it is conceivable that there is some subspace on which we can simultaneously bisect more than $k$ images of mass assignments.
We will show that this is indeed the case, even for the more general notion of center transversals:

\begin{restatable}{theorem}{center}
\label{Thm:CenterTransversalsSubspaces}
Let $\mu_1,\ldots,\mu_{n+d-k}$ be mass assignments on $G_k(\mathbb{R}^d)$, where $n\leq k\leq d$.
Then there exists a $k$-dimensional linear subspace $h$ such that $\mu_1^h,\ldots, \mu_{n+d-k}^h$ have a common $(n-1,k)$-center transversal.
\end{restatable}

In particular, for $k=n$ we get that there is always a subspace on which we can simultaneously bisect $d$ images of mass assignments.
This result will only be proved in Section \ref{Sec:transversal}.
First we will look at a conjecture by Barba \cite{Luis} which motivated this generalization:
Let $\ell$ and $\ell'$ be two lines in $\mathbb{R}^3$ in general position.
We say that $\ell$ is above $\ell'$ if the unique vertical line that intersects both $\ell$ and $\ell'$ visits first $\ell$ and then $\ell'$ when traversed from top to bottom.

\begin{conjecture}
\label{Conj:luis}
Given three sets $R, B$ and $G$ of lines in $\mathbb{R}^3$ in general position, each with an even number of lines, there is a line $\ell$ in $\mathbb{R}^3$ such that $\ell$ lies below exactly $|R|/2$ lines of $R$, $|B|/2$ lines of $B$ and $|G|/2$ lines of $G$.
That is, there is some \emph{Ham-Sandwich line} that simultaneously bisects (with respect to above-below relation) the lines of $R, B$ and $G$.
\end{conjecture}

It should be mentioned that Barba et al. have shown that the analogous statement for four sets of lines is false \cite{Luis}.
The conjecture can also be phrased in a slightly different terminology: Given three sets $R, B$ and $G$ of lines in $\mathbb{R}^3$ in general position, each with an even number of lines, there is a vertical plane $h$ such that $R\cap h$, $B\cap h$ and $G\cap h$ can be simultaneously bisected by a line in $h$.
Here, $h$ is not restricted to contain the origin, but it is restricted to be vertical, i.e., it has to be parallel to the $z$-axis.
We will prove a stronger statement of this conjecture by showing that $h$ can always be chosen to contain the origin.

More generally, in the setting of mass assignments, we show that at the cost of some masses, we can always fix $k-1$ vectors in the considered subspaces.
Without loss of generality, we assume that these vectors are vectors of the standard basis of $\mathbb{R}^d$.
We say that a linear subspace of $\mathbb{R}^d$ is \emph{$m$-horizontal}, if it contains $e_1,\ldots,e_m$, where $e_i$ denotes the $i$'th unit vector of $\mathbb{R}^d$, and we denote the space of all $m$-horizontal, $k$-dimensional subspaces of $\mathbb{R}^d$ by $Hor_k^m(\mathbb{R}^d)$.

\begin{restatable}{theorem}{hamsand}
\label{Thm:HamSandwichSubspaces}
Let $\mu_1,\ldots,\mu_{d-k+2}$ be mass assignments on $Hor_k^{k-1}(\mathbb{R}^d)$, where $2\leq k\leq d$.
Then there exists a $k$-dimensional $(k-1)$-horizontal linear subspace $h$ where $\mu_1^h,\ldots, \mu_{d-k+2}^h$ have a common Ham-Sandwich cut.
\end{restatable}

This result will be proved in Section \ref{Sec:horizontal}.
The proof of Conjecture \ref{Conj:luis} follows, after some steps to turn the lines into mass assignments, from the case $d=3$ and $k=2$.
This will be made explicit in Section \ref{Sec:lines}.

The second generalization of the Ham-Sandwich theorem that we investigate in this paper considers bisections with several cuts, where the masses are distributed into two parts according to a natural 2-coloring of the induced arrangement.
More precisely, let $\mathcal{L}$ be a set of oriented hyperplanes.
For each $\ell\in\mathcal{L}$, let $\ell^+$ and $\ell^-$ denote the positive and negative side of $\ell$, respectively (we consider the sign resulting from the evaluation of a point in these sets into the linear equation defining $\ell$).
For every point $p\in\mathbb{R}^d$, define $\lambda(p):=|\{\ell\in\mathcal{L}\mid p\in\ell^+\}|$ as the number of hyperplanes that have $p$ in their positive side.
Let $R^+:=\{p\in\mathbb{R}^d\mid \lambda(p) \text{ is even}\}$ and $R^-:=\{p\in\mathbb{R}^d\mid \lambda(p) \text{ is odd}\}$.
More intuitively, this definition can also be understood the following way: if $C$ is a cell in the hyperplane arrangement induced by $\mathcal{L}$, and $C'$ is another cell sharing a facet with $C$, then $C$ is a part of $R^+$ if and only if $C'$ is a part of $R^-$.
See Figure \ref{Fig:Checkers} for an example.
A similar setting, where the directions of the hyperplanes are somewhat restricted, has been studied by several authors \cite{Chess1, Chess3, Chess2}.

\begin{figure}
\centering
\includegraphics[scale=0.7]{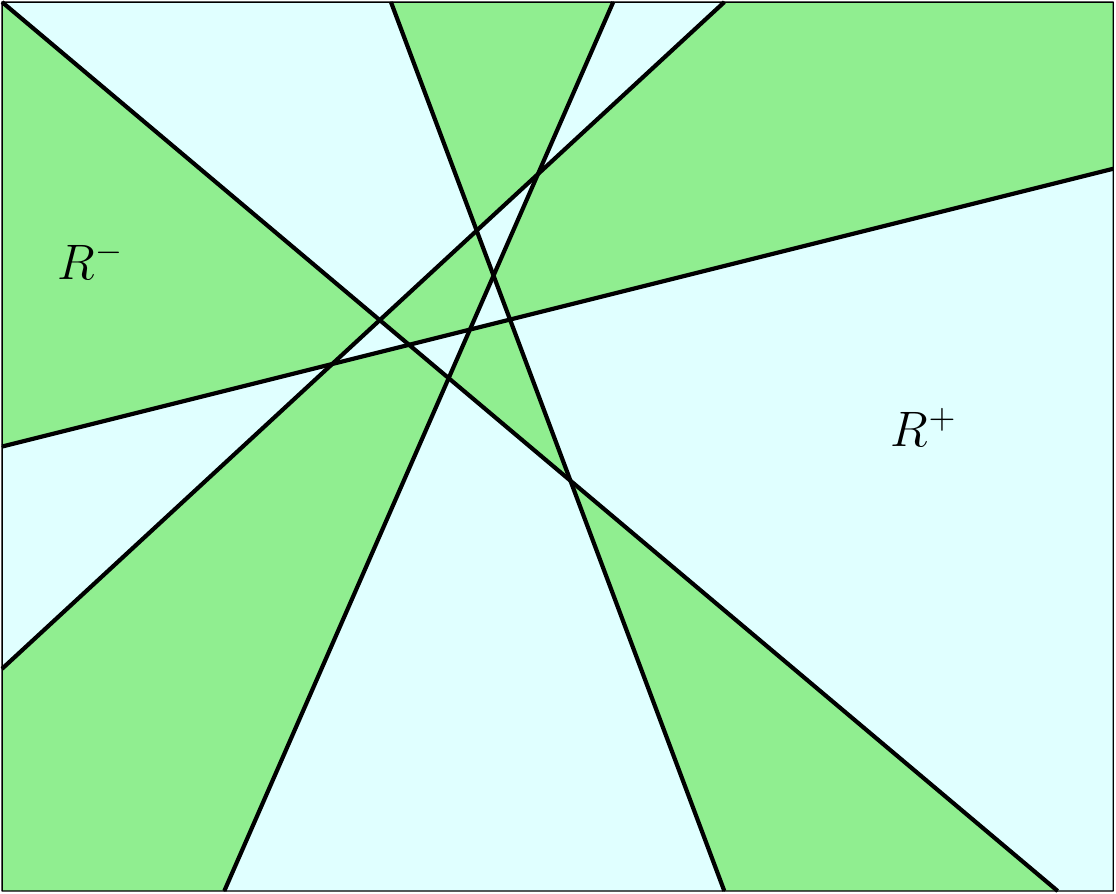}
\caption{The regions $R^+$ (light blue) and $R^-$ (green).}
\label{Fig:Checkers}
\end{figure}

We say that $\mathcal{L}$ \emph{bisects} a mass distribution $\mu$ if $\mu(R^+)=\mu(R^-)$.
Note that reorienting one hyperplane just maps $R^+$ to $R^-$ and vice versa.
In particular, if a set $\mathcal{L}$ of oriented hyperplanes simultaneously bisects a family of mass distributions $\mu_1,\ldots,\mu_k$, then so does any set $\mathcal{L'}$ of the same hyperplanes with possibly different orientations.
Thus we can ignore the orientations and say that a set $\mathcal{L}$ of (undirected) hyperplanes simultaneously bisects a family of mass distributions if some orientation of the hyperplanes does.
Langerman \cite{Stefan} conjectured the following:

\begin{conjecture}
\label{Conj:langerman}
Any $dn$ mass distributions in $\mathbb{R}^d$ can be simultaneously bisected by $n$ hyperplanes.
\end{conjecture}

For $n=1$, this is again the Ham-Sandwich theorem.
For $d=1$, this conjecture is also true, this result is known as the \emph{Necklace splitting theorem} \cite{Hobby, Matousek}.
Recently, the conjecture has been proven for several values of $n$ and $d$ \cite{pizza_cccg, Bereg, pizza2, pizza1}, but it is still open in its full generality.
In this work, we will not prove this conjecture, but we will consider a relaxed version of it:
We say that $\mathcal{L}$ \emph{almost bisects} $\mu$ if there is an $\ell\in\mathcal{L}$ such that $\mathcal{L}\setminus\{\ell\}$ bisects $\mu$.
For a family of mass distributions $\mu_1,\ldots,\mu_k$ we say that $\mathcal{L}$ \emph{almost simultaneously bisects} $\mu_1,\ldots,\mu_k$ if for every $i\in\{1,\ldots,k\}$ $\mathcal{L}$ either bisects or almost bisects $\mu_i$.
See Figure \ref{Fig:almost_bisect} for an illustration.
In this relaxed setting, we are able to prove the following:

\begin{restatable}{theorem}{pizza}
\label{Thm:almost_pizza}
Let $\mu_1,\ldots,\mu_{dn}$ be $dn$ mass distributions in $\mathbb{R}^d$. Then there are $n$ hyperplanes that almost simultaneously bisect $\mu_1,\ldots,\mu_{dn}$.
\end{restatable}

\begin{figure}
\centering
\includegraphics[scale=0.7]{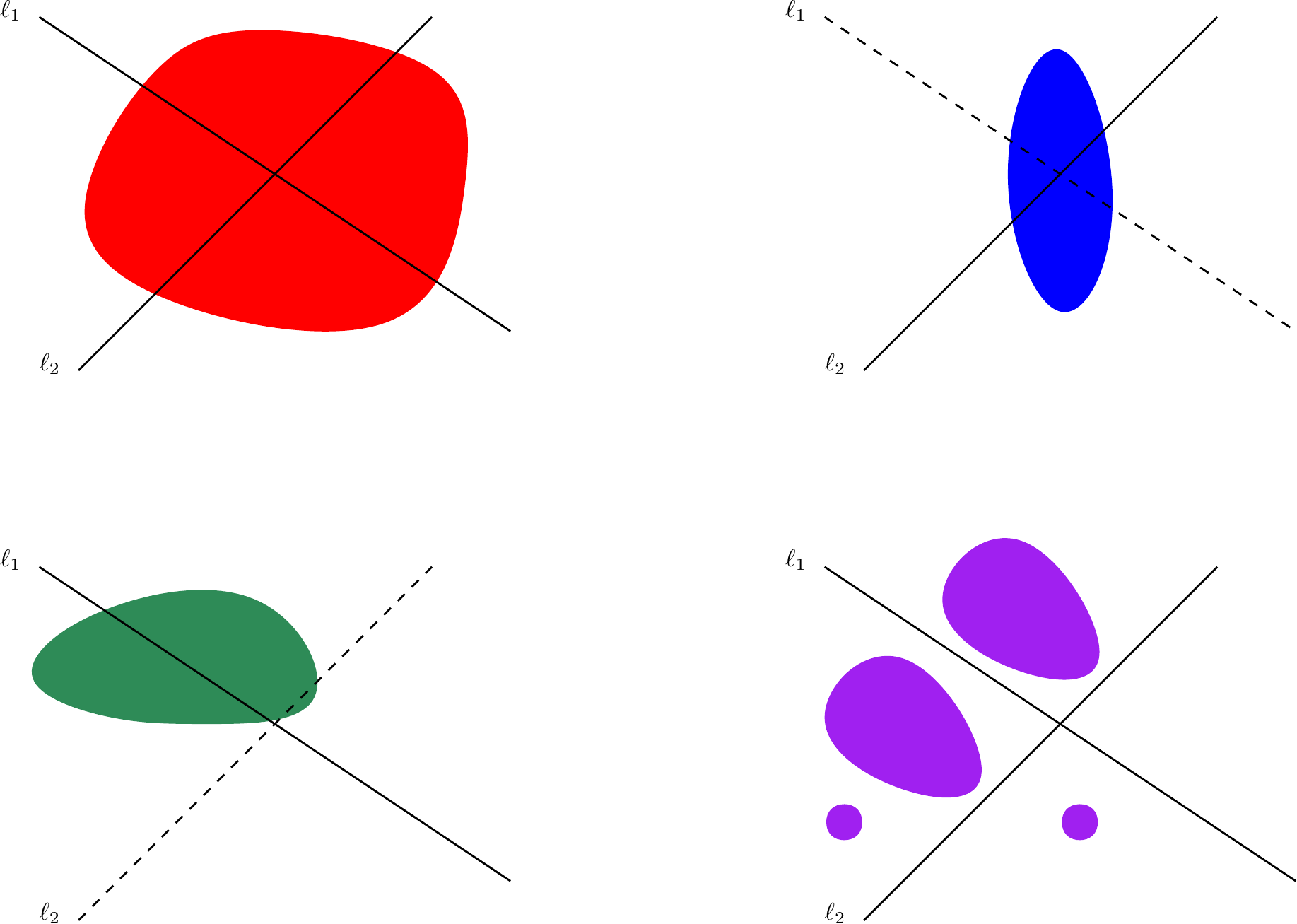}
\caption{The lines $\ell_1$ and $\ell_2$ almost simultaneously bisect four masses.}
\label{Fig:almost_bisect}
\end{figure}

We hope that our methods might extend to a proof of Conjecture \ref{Conj:langerman}.
We will first prove a similar result where we enforce that all bisecting hyperplanes contain the origin.
The general version then follows from lifting the problem one dimension higher.
The proof is based on the following idea: for each mass, $n-1$ of the hyperplanes define two regions, one we take with positive sign, the other with negative sign.
This defines a so called \emph{charge} (a mass distribution, which unfortunately is locally negative, which is why we will need the relaxed setting).
The $n$'th hyperplane should now bisect this new mass distribution.
However, this $n$'th hyperplane now again changes the other mass distributions, so in the end we want to guarantee that there are $n$ hyperplanes such that all of them correctly bisect the masses.
More precisely, let $G_{d-1}(\mathbb{R}^d)^n$ be the space of all sets of $n$ hyperplanes containing the origin (i.e., linear subspaces) in $\mathbb{R}^d$.
Similar to before, we define a mass assignment $\mu$ on $G_{d-1}(\mathbb{R}^d)^n$ as a continuous assignment $G_{d-1}(\mathbb{R}^d)^n\rightarrow M_d$, where $M_d$ again denotes the space of all $d$-dimensional mass distributions.
In other words, $\mu$ continuously assigns a mass distribution $\mu^p:=\mu(p)$ to $\mathbb{R}^d$ for each $p=(h_1,\ldots,h_n)\in Gr_{d-1}(\mathbb{R}^d)^n$.
An example of such mass assignments could be the intersection of a fixed $d$-dimensional mass distribution with the Minkowski sum of the hyperplanes with a unit ball.
In Section \ref{Sec:product}, we will prove the following:

\begin{restatable}{theorem}{product}
\label{Thm:product}
Let $\mu_1,\ldots,\mu_{(d-1)n}$ be $(d-1)n$ mass assignments on $G_{d-1}(\mathbb{R}^d)^n$.
Then there exists $p=(h_1,\ldots,h_n)\in  Gr_{d-1}(\mathbb{R}^d)^n$ such that for every $i\in\{1,\ldots,n\}$, the hyperplane $h_i$ simultaneously bisects $\mu_{(d-1)(i-1)+1}^p,\ldots,\mu_{(d-1)i}^p$.
\end{restatable}

We then use the underlying topological result to prove Theorem \ref{Thm:almost_pizza} in Section \ref{Sec:pizza}.
All the results are proved using topological methods, and the underlying topological results might be of independent interest.
For an introduction to topological methods, we refer to the books by Matou{\v{s}}ek \cite{Matousek} and de Longueville \cite{Longueville}.
Most of the proofs in this work use so-called Stiefel-Whitney classes of vector bundles.
The standard reference for this concept is the classic book by Milnor and Stasheff \cite{milnor}.

\section{Ham Sandwich Cuts in horizontal subspaces}
\label{Sec:horizontal}

In order to prove Theorem \ref{Thm:HamSandwichSubspaces}, we establish a few preliminary lemmas.
Consider the following space, which we denote by $F_{hor}$: the elements of $F_{hor}$ are pairs $(h,\overrightarrow{\ell})$, where $h$ is an (unoriented) $k$-dimensional $(k-1)$-horizontal linear subspace of $\mathbb{R}^d$ and $\overrightarrow{\ell}$ is an oriented $1$-dimensional linear subspace of $h$, that is, an oriented line in $h$ through the origin.
The space $F_{hor}$ inherits a topology from the Stiefel manifold.
Furthermore, inverting the orientation of $\overrightarrow{\ell}$ is a free $\mathbb{Z}_2$-action, giving $F_{hor}$ the structure of a $\mathbb{Z}_2$-space.

We will first give a different description of the space $F_{hor}$.
Define
\[ F':=\faktor{S^{d-k}\times S^{k-2}\times[0,1]}{(\approx_0,\approx_1)},\]
where $(x,y,0)\approx_0 (x,y',0)$ for all $y,y'\in S^{k-2}$ and $(x,y,1)\approx_1 (-x,y,1)$ for all $x\in S^{d-k}$.
Further, define a free $\mathbb{Z}_2$-action on $F'$ by $-(x,y,t):=(-x,-y,t)$.
We claim that the $\mathbb{Z}_2$-space $F'$ is "the same" as $F_{hor}$:

\begin{lemma}
\label{Lem:spaces}
There is a $\mathbb{Z}_2$-equivariant homeomorphism between $F'$ and $F_{hor}$.
\end{lemma}

\begin{proof}
Consider the subspace $Y\subset\mathbb{R}^d$ spanned by $e_1,\ldots,e_{k-1}$.
The space of unit vectors in $Y$ is homeomorphic to $S^{k-2}$.
Similarly let $X\subset\mathbb{R}^d$ be spanned by $e_k,\ldots,e_{d}$.
Again, the space of unit vectors in $X$ is homeomorphic to $S^{d-k}$.
In a slight abuse of notation, we will write $y$ and $x$ both for a unit vector in $Y$ and $X$ as well as for the corresponding points in $S^{k-2}$ and $S^{d-k}$, respectively.

We first construct a map $\varphi$ from $S^{d-k}\times S^{k-2}\times[0,1]$ to $F_{hor}$ as follows: for every $x\in S^{d-k}$ let $h(x)$ be the unique $(k-1)$-horizontal subspace spanned by $x,e_1,\ldots,e_{k-1}$.
See Figure \ref{Fig:space_homeo} for an illustration.
Note that $h(-x)=h(x)$.
Further, define $v(x,y,t):=(1-t)x+ty$ and let $\overrightarrow{\ell}(x,y,t)$ be the directed line defined by the vector $v(x,y,t)$.
Note that $\overrightarrow{\ell}(x,y,t)$ lies in the plane spanned by $x$ and $y$ and thus also in $h(x)$.
Finally, set $\varphi(x,y,t):=(h(x),\overrightarrow{\ell}(x,y,t))$.
Both $h$ and $v$ are both open and closed continuous maps, and thus so is $\varphi$.
Also, we have that $v(-x-y,t)=-(1-t)x-ty=-v(x,y,t)$, so $\varphi$ is $\mathbb{Z}_2$-equivariant.

Note that for $t=0$ we have $v(x,y,0)=x$, so $\varphi(x,y,0)$ does not depend on $y$, and in particular $\varphi(x,y,0)=\varphi(x,y',0)$ or all $y,y'\in S^{k-2}$.
Similarly, for $t=1$ we have $v(x,y,1)=y$ and $h(-x)=h(x)$, and thus $\varphi(x,y,1)=\varphi(-x,y,1)$ for all $x\in S^{d-k}$.
Hence, $\varphi$ induces a map $\varphi'$ from $F'$ to $F_{hor}$ which is still open, closed, continuous and $\mathbb{Z}_2$-equivariant.
Finally, it is easy to see that $\varphi'$ is bijective.
Thus, $\varphi'$ is a $\mathbb{Z}_2$-equivariant homeomorphism between $F'$ and $F_{hor}$, as required.
\end{proof}

\begin{figure}
\centering
\includegraphics[scale=0.7]{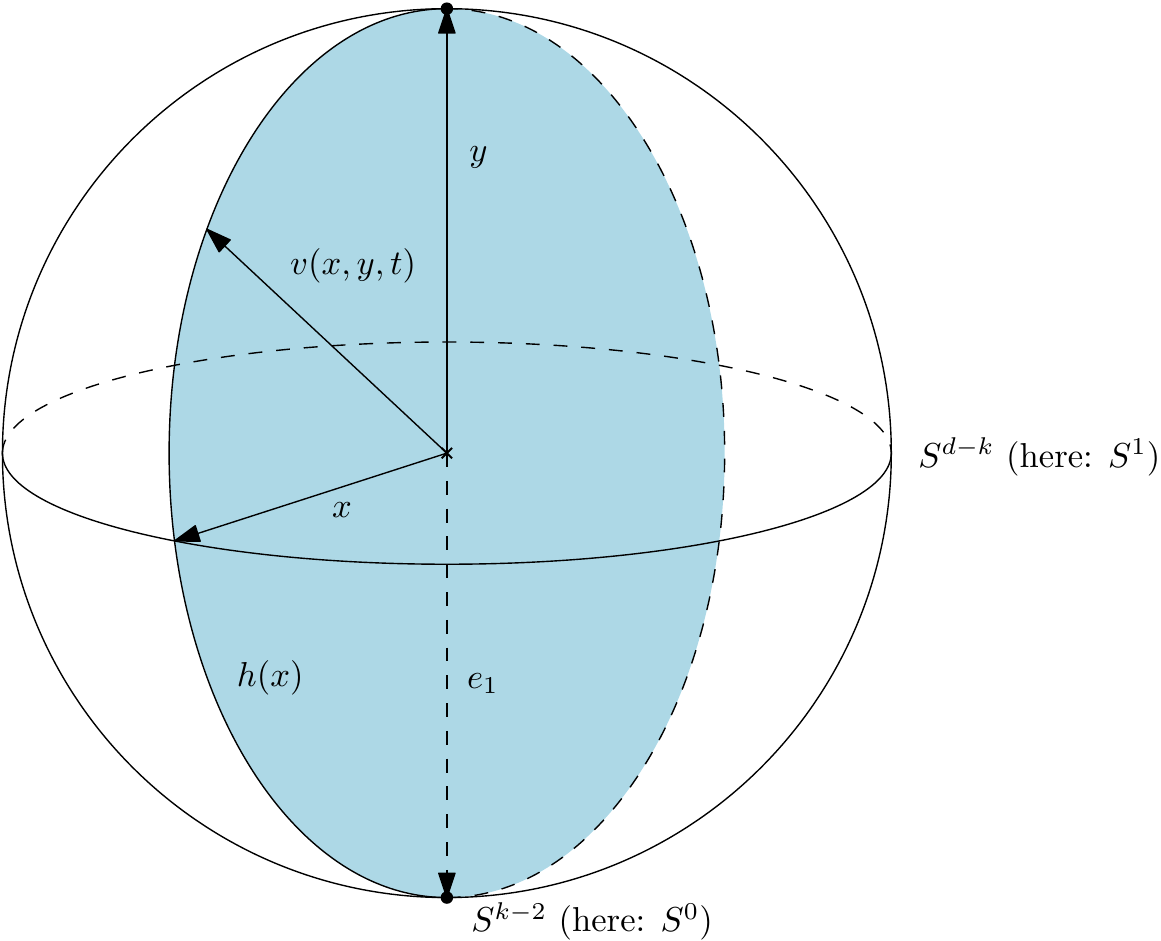}
\caption{The map $\varphi$ for $d=3$ and $k=2$.}
\label{Fig:space_homeo}
\end{figure}

We now prove a Borsuk-Ulam-type statement for $F_{hor}$.

\begin{lemma}
\label{Lem:butype}
There is no $\mathbb{Z}_2$-map $f: F_{hor} \rightarrow S^{d-k}$.
\end{lemma}

\begin{proof}
Assume for the sake of contradiction that $f$ exists.
Then, by Lemma \ref{Lem:spaces}, $f$ induces a map $F: S^{d-k}\times S^{k-2}\times[0,1]\rightarrow S^{d-k}$ with the following properties:
\medskip
\begin{enumerate}
\item[(1)] $F(-x,-y,t)=-F(x,y,t)$ for all $t\in (0,1)$;
\item[(2)] $F(x,y,0)=F(x,y',0)$ for all $y,y'\in S^{k-2}$ and $F(-x,y,0)=-F(x,y,0)$ for all $x\in S^{d-k}$;
\item[(3)] $F(x,-y,1)=-F(x,y,1)$ for all $y\in S^{k-2}$ and $F(-x,y,1)=F(x,y,1)$ for all $x\in S^{d-k}$.
\end{enumerate}
\medskip
In particular, $F$ is a homotopy between $f_0(x,y):=F(x,y,0)$ and $f_1(x,y):=F(x,y,1)$.
Fix some $y_0\in S^{k-2}$.
Then $F$ induces a homotopy between $g_0(x):=f_0(x,y_0)$ and $g_1(x):=f_1(x,y_0)$.
Note that $g_0: S^{d-k}\rightarrow S^{d-k}$ has odd degree by property (2).
On the other hand, $g_1: S^{d-k}\rightarrow S^{d-k}$ has even degree by property (3).
Thus, $F$ induces a homotopy between a map of odd degree and a map of even degree, which is a contradiction.
\end{proof}

We now have all tools that are necessary to prove Theorem \ref{Thm:HamSandwichSubspaces}.

\hamsand*

\begin{proof}
For each $\mu_i$ and $(h,\overrightarrow{\ell})$, consider the point $v_i$ on $\overrightarrow{\ell}$ for which the orthogonal hyperplane bisects $\mu_i^h$.
(If $v_i$ is not unique, the set of all possible such points is an interval, in which case we choose $v_i$ as the midpoint of this interval.)
This induces a continuous $\mathbb{Z}_2$-map $g:F_{hor}\rightarrow\mathbb{R}^{d-k+2}$.
For $i\in\{1\ldots,d-k+1\}$, set $w_i:=v_i-v_{d-k+2}$.
The $w_i$'s then induce a continuous $\mathbb{Z}_2$-map $f:F_{hor}\rightarrow\mathbb{R}^{d-k+1}$.
We want to show that there exists $(h,\overrightarrow{\ell})$ where $v_1=v_2=\ldots =v_{d-k+2}$, or equivalently, $w_1=\ldots,w_{d-k+1}=0$, i.e., $f$ has a zero.
Assume that this is not the case.
Then normalizing $f$ induces a $\mathbb{Z}_2$-map $f':F_{hor}\rightarrow S^{d-k}$, which is a contradiction to Lemma \ref{Lem:butype}.
\end{proof}

Note that the higher $k$ is chosen, the weaker our result.
In fact, for $k>\frac{d}{2}+1$, our result is weaker than what we would get from the Ham-Sandwich theorem.
We conjecture that this trade-off is not necessary:
\begin{conjecture}
Let $\mu_1,\ldots,\mu_{d}$ be mass assignments in $\mathbb{R}^d$ and $k\geq 2$.
Then there exists a $k$-dimensional $(k-1)$-horizontal linear subspace $h$ such that $\mu_1^h,\ldots, \mu_{d}^h$ have a common Ham-Sandwich cut.
\end{conjecture}

\section{Application: bisecting lines in space}
\label{Sec:lines}

Recall the setting of Conjecture \ref{Conj:luis}: Given three sets $R, B$ and $G$ of lines in $\mathbb{R}^3$ in general position, each with an even number of lines, is there a line $\ell$ in $\mathbb{R}^3$ such that $\ell$ lies below exactly $|R|/2$ lines of $R$, $|B|/2$ lines of $B$ and $|G|/2$ lines of $G$?
Here, general position means that (i) no two lines are parallel,  (ii) no line is vertical (i.e., parallel to the $z$-axis), (iii) no line intersects the $z$-axis and (iv) for any four lines, if there is a line intersecting all of them, the (unique) vertical plane containing this common intersecting line does not go through the origin.

We want to prove that there always is such a line $\ell$ using Theorem \ref{Thm:HamSandwichSubspaces}.
In order to apply Theorem \ref{Thm:HamSandwichSubspaces}, we need to define a mass assignment.
To this end, we replace every line $r$ in $R$ by a very thin infinite cylinder of radius $\varepsilon$, centered at $r$.
Denote the collection of cylinders obtained this way by $R^*$.
Define $B^*$ and $G^*$ analogously.
For each vertical plane $h$ through the origin, let $D_K^h$ be a disk in $h$ centered at the origin, with some (very large) radius $K$.
Define $\mu_R^h$ as $(R^*\cap h)\cap D_K^h$.
It is straightforward to show that $\mu_R^h$ is a mass assignment.
Analogously we can define mass assignments $\mu_B^h$ and $\mu_G^h$.
From Theorem \ref{Thm:HamSandwichSubspaces}, where we set $e_1$ to be the unit vector on the $z$-axis, we deduce that there is a vertical plane $h_0$ and a line $\ell\in h_0$ such that $\ell$ simultaneously bisects $\mu_R^{h_0}$, $\mu_B^{h_0}$ and $\mu_G^{h_0}$.
We claim that this $\ell$ gives a solution to Conjecture \ref{Conj:luis}.

To show this, we distinguish two cases:
The first case is that all the cylinders in $R^* \cup B^* \cup G^*$ intersect $D_K^{h_0}$.
In this case, it is a standard argument to show that $\ell$ is a Ham-Sandwich cut of the point set $(R \cup B \cup G)\cap h_0$.
Note that because of general position assumptions (ii) and (iv), at most one triple of points in $(R \cup B \cup G)\cap h_0$ is collinear.
As all three sets have an even number of lines, we thus have that $\ell$ either contains two points or no point at all.
Further, if it contains two points $p_1$ and $p_2$, then they must have the same color.
In this case, slightly rotate $\ell$ such that $p_1$ lies above $\ell$ and $p_2$ lies below $\ell$.
Now, in any case, $\ell$ is a Ham-Sandwich cut that contains no points.
In particular, $\ell$ lies below exactly $|R|/2$ lines of $R$, $|B|/2$ lines of $B$ and $|G|/2$ lines of $G$, which is what we required.

The second case is that some cylinders in $R^* \cup B^* \cup G^*$ do not intersect $D_K^{h_0}$.
By the general position assumption (i), choosing $K$ sufficiently large, we can assume that exactly one cylinder $c^*$ does not intersect $D_K^{h_0}$.
Without loss of generality, let $c^*\in R^*$, defined by some line $c\in R$.
If $K$ is chosen sufficiently large, by general position assumption (iii) we can further assume that $c$ is parallel to $h_0$.
Thus, similar to above, $\ell$ is a Ham-Sandwich cut of the point set $((R\setminus\{c\}) \cup B \cup G)\cap h_0$.
Again, at most one triple of points is collinear.
As $(R\setminus\{c\})$ contains an odd number of lines, $\ell$ passes through either 1 or 3 points.
If $\ell$ passes through 3 points $p_1$, $p_2$ and $p_3$, then without loss of generality $p_1\in (R\setminus\{c\})\cap h_0$.
Further, $p_2$ and $p_3$ must be induced by the same set of lines, without loss of generality $B$.
In both cases, we can slightly rotate $\ell$ such that $p_1$ is above $\ell$ and $p_2$ and $p_3$ lie on different sides of $\ell$.
Similarly, if $\ell$ contains 1 point $p_1$, then $p_1\in (R\setminus\{c\})\cap h_0$, and we can slightly translate $\ell$ such that $p_1$ lies above $\ell$.
Now again, $\ell$ lies below exactly $|R|/2$ lines of $R$, $|B|/2$ lines of $B$ and $|G|/2$ lines of $G$, which is what we required.

Thus, we have proved the following Theorem:

\begin{theorem}
\label{Thm:lines}
Given three sets $R, B$ and $G$ of lines in $\mathbb{R}^3$ in general position, each with an even number of lines, there is a line $\ell$ in $\mathbb{R}^3$ such that $\ell$ lies below exactly $|R|/2$ lines of $R$, $|B|/2$ lines of $B$ and $|G|/2$ lines of $G$.
\end{theorem}

\section{Center Transversals in general subspaces}
\label{Sec:transversal}

In this section we consider the more general case of assignments of mass distributions to all linear subspaces.
The space of all linear subspaces of fixed dimension defines in a natural way a \emph{vector bundle}.
Recall the following definition: a vector bundle consists of a base space $B$, a total space $E$, and a continuous projection map $\pi: E\mapsto B$.
Furthermore, for each $b\in B$, the fiber $\pi^{-1}(b)$ over $b$ has the structure of a vector space over the real numbers.
Finally, a vector bundle satisfies the \emph{local triviality condition}, meaning that for each $b\in B$ there is a neighborhood $U\subset B$ containing $p$ such that $\pi^{-1}(U)$ is homeomorphic to $U\times\mathbb{R}^d$.
A \emph{section} of a vector bundle is a continuous mapping $s: B\mapsto E$ such that $\pi s$ equals the identity map, i.e., $s$ maps each point of $B$ to its fiber.
Recall that we denote by $G_m(\mathbb{R}^n)$ the Grassmann manifold consisting of all $m$-dimensional subspaces of $\mathbb{R}^n$.
Let $\gamma_m^d$ be the \emph{canonical bundle} over $G_m(\mathbb{R}^n)$.
The bundle $\gamma_m^n$ has a total space $E$ consisting of all pairs $(L,v)$, where $L$ is an $m$-dimensional subspace of $\mathbb{R}^n$ and $v$ is a vector in $L$, and a projection $\pi: E\mapsto G_m(\mathbb{R}^n)$ given by $\pi((L,v))=L$.
Another space that we will be working with is the \emph{complete flag manifold} $\tilde{V}_{n,n}$: a \emph{flag} $\mathcal{F}$ in a vector space $V$ of dimension $n$ is an increasing sequence of subspaces of the form
\[ \mathcal{F}=\{0\}=V_0\subset V_1\subset\cdots\subset V_k=V. \]
A flag is a \emph{complete flag} if $\text{dim}V_i=i$ for all $i$ (and thus $k=n$).
The complete flag manifold $\tilde{V}_{n,n}$ is the manifold of all complete flags of $\mathbb{R}^n$.
Similar to the Grassmann manifold, we can define a canonical bundle for each $V_i$, which we will denote by $\vartheta_i^n$.
For details on vector bundles and sections, see \cite{milnor}.

\begin{lemma}
\label{Lem:Sections}
Let $s_1,\ldots,s_{m+1}$ be $m+1$ sections of the canonical bundle $\vartheta_l^{m+l}$.
Then there is a flag $\mathcal{F}\in\tilde{V}_{m+l,m+l}$ such that $s_1(\mathcal{F})=\ldots=s_{m+1}(\mathcal{F})$.
\end{lemma}

This Lemma is a generalization of Proposition 2 in \cite{zivaljevic} and Lemma 1 in \cite{dolnikov}.
Our proof follows the proof in \cite{zivaljevic}.

\begin{proof}
Consider the sections $q_i:=s_{m+1}-s_i$.
We want to show that there exists a flag $\mathcal{F}$ for which $q_1(\mathcal{F})=\ldots =q_m(\mathcal{F})=0$.
The sections $q_1,\ldots, q_m$ determine a unique section in the $m$-fold Whitney sum of $\vartheta_l^{m+l}$, which we denote by $W$.
Note that $W$ has base $\tilde{V}_{m+l,m+l}$ and fiber dimension $ml$.
We will show that $W$ does not admit a nowhere zero section.
For this, it suffices to show that the highest Stiefel-Whitney class $w_{ml}(W)$ is nonzero (see \cite{milnor}, \S 4, Proposition 3).

By the Whitney product formula we have $w_{ml}(W)=w_l(\vartheta_l^{m+l})^m$.
Note that the projection $f:\tilde{V}_{m+l,m+l}\rightarrow G_l(\mathbb{R}^{m+l})$ which maps $(V_0,\ldots,V_l,\ldots,V_{m+l})$ to $V_l$ induces a bundle map from $\vartheta_l^{m+l}$ to $\gamma_l^{m+l}$.
Thus by the naturality of Stiefel-Whitney classes we have $w_l(\vartheta_l^{m+l})^m=f^*(w_l(\gamma_l^{m+l})^m)=f^*(w_l(\gamma_l^{m+l}))^m$.
Further, we have the following commutative diagram
\[
  \begin{tikzcd}
    \tilde{V}_{m+l,m+l} \arrow{r}{i} \arrow[swap]{d}{f} & \tilde{V}_{\infty,m+l} \arrow{d}{g} \\
    G_l(\mathbb{R}^{m+l}) \arrow[swap]{r}{j} & G_l(\mathbb{R}^{\infty})
  \end{tikzcd}
,\]
where $i$ and $j$ are inclusions and $g$ is the canonical map from $\tilde{V}_{\infty,m+l}$ to $G_l(\mathbb{R}^{\infty})$ (see e.g. \cite{hausmann, zivaljevic}).
In $\mathbb{Z}_2$-cohomology, we get the following diagram:
\[
  \begin{tikzcd}
   H^*( \tilde{V}_{m+l,m+l}) & H^*(\tilde{V}_{\infty,m+l}) \arrow[swap]{l}{i^*} \\
    H^*(G_l(\mathbb{R}^{m+l})) \arrow{u}{f^*} & H^*(G_l(\mathbb{R}^{\infty})) \arrow{l}{j^*} \arrow{u}[swap]{g^*}
  \end{tikzcd}
.\]
It is known that $H^*(\tilde{V}_{\infty,m+l})$ is a polynomial algebra $\mathbb{Z}_2[t_1,\ldots,t_{m+l}]$ and that $g^*$ maps $H^*(G_l(\mathbb{R}^{\infty}))$ injectively onto the algebra $\mathbb{Z}_2[\sigma_1,\ldots,\sigma_l]\subset\mathbb{Z}_2[t_1,\ldots,t_l]\subset\mathbb{Z}_2[t_1,\ldots,t_{m+l}]$, where $\sigma_i$ denotes the $i$'th symmetric polynomial in the variables $t_1,\ldots,t_l$ \cite{milnor, hausmann}.
Further, $H^*( \tilde{V}_{m+l,m+l})$ is a polynomial algebra $\faktor{\mathbb{Z}_2[t_1,\ldots,t_{m+l}]}{(\sigma_1,\ldots,\sigma_{m+l})}$ and $i^*$ is the corresponding quotient map.
Since $w_l({\gamma_l^{m+l}})=j^*(\sigma_l)$, we have $w_l(\vartheta_l^{m+l})^m=f^*(j^*(\sigma_l))^m$ and in particular $w_l(\vartheta_l^{m+l})^m=0$ would imply that $(\sigma_l)^m\in\ker i^*$, i.e. $(t_1\cdots t_l)^m$ is in the ideal $(\sigma_1,\ldots,\sigma_{m+l})$.
But this is a contradiction to Proposition 2.21 in \cite{zivguide}.
\end{proof}

Consider now a continuous map $\mu: \tilde{V}_{m+l,m+l}\rightarrow M_l$, which assigns an $l$-dimensional mass distribution to $V_l$ for every flag.
We call such a map an \emph{$l$-dimensional mass assignment on $\tilde{V}_{m+l,m+l}$}.

\begin{corollary}
\label{Cor:GeneralMasses}
Let $\mu_1,\ldots,\mu_{m+1}$ be $l$-dimensional mass assignments on $\tilde{V}_{m+l,m+l}$.
Then there exists a flag $\mathcal{F}\ni V_l$ such that some point $p\in V_l$ is a centerpoint for all $\mu_1^{\mathcal{F}},\ldots,\mu_{m+1}^{\mathcal{F}}$. 
\end{corollary}

\begin{proof}
For every $\mu_i$ and every flag $\mathcal{F}$, the centerpoint region of $\mu_i^{\mathcal{F}}$ is a convex compact region in the respective $V_l$.
In particular, for each $\mu_i$ we get a multivalued, convex, compact section $s_i$ in $\vartheta_l^{m+l}$.
Using Proposition 1 from \cite{zivaljevic}, Lemma \ref{Lem:Sections} implies that there is a Flag in which all $s_i$'s have a common point $p$.
\end{proof}

From this we can now deduce Theorem \ref{Thm:CenterTransversalsSubspaces}

\center*

\begin{proof}
Note that a $(n-1,k)$-center transversal in a $k$-dimensional space is a common centerpoint of the projection of the masses to a $k-(n-1)$-dimensional subspace.
Consider a flag $\mathcal{F}=(V_0,\ldots,V_d)$.
For each mass assignment $\mu_i$ define $\mu'_i(\mathcal{F}):=\pi_{k-(n-1)}(\mu_i^{V_k})$, where $\pi_{k-(n-1)}$ denotes the projection from $V_k$ to $V_{k-(n-1)}$.
Every $\mu'_i$ is an $(k-(n-1))$-dimensional mass assignment on $\tilde{V}_{d,d}$.
The result now follows from Corollary \ref{Cor:GeneralMasses} by setting $l=k-(n-1)$ and $m=d-k+n-1$.
\end{proof}

\section{Sections in product bundles}
\label{Sec:product}

Similar to before, we again work with vector bundles, but now over a different space.
Recall that a mass assignment $\mu$ on $G_{d-1}(\mathbb{R}^d)^n$ assigns a $d$-dimensional mass distribution $\mu^p$ to $\mathbb{R}^d$ for each $p=(h_1,\ldots,h_n)\in Gr_{d-1}(\mathbb{R}^d)^n$.
We want to show that given  $(d-1)n$ such mass assignments, there is a $p$ such that each $h_i$ bisects $d-1$ of their images.
The idea is the following: we assign $d-1$ masses to each $h_i$.
For every $p$, we now sweep a copy of $h_i$ along a line $\ell$ orthogonal to $h_i$ and for every mass assigned to $h_i$ we look at the point on $\ell$ for which the swept copy through that point bisects the mass.
We want to show that for some $p$, all these points coincide with the origin.

\begin{lemma}
\label{lem:product}
Consider the vector bundle $\xi:=(\gamma_{m}^{d})^k$ (the $k$-fold Cartesian product of $\gamma_{m}^{d}$) over the space $B:=G_{m}(\mathbb{R}^{d})^k$.
Let $q:=d-m$.
Then for any $q$ sections $s_1,\ldots, s_q$ of $\xi$ there exists $b\in B$ such that $s_1(b)=\ldots =s_q(b)=0$.
\end{lemma}

This Lemma is another generalization of Proposition 2 in \cite{zivaljevic} and Lemma 1 in \cite{dolnikov}.
Our proof follows the proof in \cite{dolnikov}.

\begin{proof}
The sections $s_1,\ldots, s_q$ determine a unique section in the $q$-fold Whitney sum of $\xi$, which we denote by $\xi^q$.
$\xi^q$ has base $B$ and fiber dimension $kqm$.
We want to show that $\xi^q$ does not allow a nowhere zero section.
For this, it is again enough to show that the highest Stiefel-Whitney class $w_{kqm}(\xi^q)$ does not vanish.
Denote by $\Gamma_m^d$ the $q$-fold Whitney sum of $\gamma_m^d$ and consider the vector bundle $\zeta:=(\Gamma_m^d)^k$.
Note that $\zeta$ also has base $B$ and fiber dimension $kqm$.
Furthermore, there is a natural bundle map from $\zeta$ to $\xi^q$, and as they have the same base space, $\zeta$ and $\xi^q$ are isomorphic (see \cite{milnor}, \S 3, Lemma 3.1).
Thus, it is enough to show that the highest Stiefel-Whitney class $w_{kqm}(\zeta)$ does not vanish.
The Stiefel-Whitney classes of a Cartesian product of vector bundles can be computed as the cross product of the Stiefel-Whitney classes of its components in the following way (see \cite{milnor}, \S 4, Problem 4-A): 
\[ w_j(\eta_1\times\eta_2)=\sum_{i=0}^j w_i(\eta_1)\times w_{j-i}(\eta_2). \]
It was shown by Dol'nikov \cite{dolnikov} that $w_{qm}(\Gamma_m^d)=1\in\mathbb{Z}_2=H^{qm}(G_m(\mathbb{R}^d);\mathbb{Z}_2)$.
By the K\"{u}nneth theorem and induction it follows that $w_{kqm}((\Gamma_m^d)^k)=1\in\mathbb{Z}_2=H^{kqm}((G_m(\mathbb{R}^d))^k;\mathbb{Z}_2)$.
\end{proof}

In the following, we will use Lemma \ref{lem:product} only for the case $m=1$, i.e., for products of line bundles.
This case could also be proved using a Borsuk-Ulam-type result on product of spheres (Theorem 4.1 in \cite{dzedzej}, for $n_1=\ldots=n_r=d-1$, see also \cite{Ramos}).
Consider now $B:=G_{1}(\mathbb{R}^{d})^n$, i.e., all $n$-tuples of lines in $\mathbb{R}^d$ through the origin.
Further, for every $i\in\{1,\ldots,n\}$ we define $\xi_i$ as the following vector bundle: the base space is $B$, the total space $E_i$ is the set of all pairs $(b,v)$, where $b=(\ell_1(b),\cdots,\ell_n(b))$ is an element of $B$ and $v$ is a vector in $\ell_i(b)$, and the projection $\pi$ is given by $\pi((b,v))=b$.
It is straightforward to show that this is indeed a vector bundle.
In other words, we consider one line to be marked and the fiber over an $n$-tuple of lines is the $1$-dimensional vector space given by the marked line.
We are now ready to prove Theorem \ref{Thm:product}.

\product*

\begin{proof}
Consider $\xi=(\gamma_{1}^{d})^n$.
Recall that $B=G_{1}(\mathbb{R}^{d})^n$.
For an element $b=(\ell_1(b),\cdots,\ell_n(b))$ of $B$, consider for every $i\in\{1,\ldots,n\}$ the $(d-1)$-dimensional hyperplane through the origin that is orthogonal to $\ell_i(b)$ and denote it by $h_i(b)$.
Similarly, for every $(b,v)\in E_i$, let $g_i(b,v)$ be the hyperplane through $v$ orthogonal to $\ell_i(b)$. Note that $g_i(b,0)=h_i(b)$.

Consider now the mass $\mu_1$.
The set of all pairs $(b,v)$ such that $(g_1(b,v), h_2(b),\ldots, h_n(b))$ bisects $\mu_1$ defines a section $s_1^1$ in $\xi_1$.
Analogously, use $\mu_{(d-1)(j-1)+1}$ to define $s_j^1$ for all $j\in\{2,\ldots,n\}$.
Then $s^1:=(s_1^1,\ldots,s_n^1)$ is a section in $(\gamma_{1}^{d})^n$.
Similarly, for $i\in\{2,\ldots,d-1\}$, using the masses $\mu_{(d-1)(j-1)+i}$ for all $j\in\{2,\ldots,n\}$ define a section $s^i$ in $(\gamma_{1}^{d})^n$.

We have thus defined $d-1$ sections in $(\gamma_{1}^{d})^n$.
Hence, by applying Lemma \ref{lem:product}, we get that there is a point $b_0$ in $B$ such that $s^1(b_0),\ldots,s^{d-1}(b_0)=0$.
In particular, all orthogonal hyperplanes $g_i(b,v)$ contain the origin, so their collection is an element of $G_{d-1}(\mathbb{R}^d)^n$.
Further, it follows from the definition of the sections $s^i$ that $h_i$ simultaneously bisects $\mu_{(d-1)(i-1)+1}^p,\ldots,\mu_{(d-1)i}^p$.
\end{proof}

\section{Application: bisections with several cuts}
\label{Sec:pizza}

The objective of this section is to prove Theorem \ref{Thm:almost_pizza}.
Before we dive into the technicalities, let us briefly discuss the main ideas.
We first show that any $(d-1)n$ mass distributions in $\mathbb{R}^d$ can be almost simultaneously bisected by $n$ hyperplanes through the origin.
The idea of this proof is very similar to the proof of Theorem \ref{Thm:product}: consider some mass $\mu$ and assume that $n-1$ of the hyperplanes are fixed.
Sweep the last hyperplane along a line through the origin and stop when the resulting arrangement of $n$ hyperplanes almost bisects $\mu$.
We do the same for every mass, one hyperplane is swept, the others are considered to be fixed.
Each hyperplane is swept for $(d-1)$ masses.
Using Lemma \ref{Lem:Sections}, we want to argue, that there is a solution, such that all the swept hyperplanes are stopped at the origin.
The only problem with this approach is, that the points where we can stop the hyperplane are in general not unique.
In fact, the region of possible solutions for one sweep can consist of several connected components, so in particular, it is not a section, and we cannot use Lemma \ref{Lem:Sections} directly.
We will therefore need another Lemma, that says that we can find find a section in this space of possible solutions.
This Lemma is actually the only reason why our approach only works for the relaxed setting: we need to sometimes ignore certain hyperplanes to construct such a section.
However, constructing a section that lies completely in the space of solutions is stronger than what we would need to use Lemma \ref{Lem:Sections}.
It would be enough to argue, that assuming no almost simultaneous bisection exists, we could find a nowhere zero section contradicting Lemma \ref{Lem:Sections}.
It is thus possible that our approach could be strengthened to prove Conjecture \ref{Conj:langerman}.

Let us now start by stating the aforementioned result for bisections with hyperplanes containing the origin:

\begin{theorem}
\label{Thm:pizza_origin}
Let $\mu_1,\ldots,\mu_{(d-1)n}$ be $(d-1)n$ mass distributions in $\mathbb{R}^d$. Then there are $n$ hyperplanes, all containing the origin, that almost simultaneously bisect $\mu_1,\ldots,\mu_{(d-1)n}$.
\end{theorem}

As mentioned, in order to prove this result, we need a few additional observations.
In the following, by a \emph{limit antipodal} function we mean a continuous function $f:\mathbb{R}\mapsto\mathbb{R}$ with the following two properties:
\begin{enumerate}
\item $\lim_{x\rightarrow\infty}f=-\lim_{x\rightarrow -\infty}f$,
\item the set of zeroes of $f$ consists of finitely many connected components.
\end{enumerate}
See Figure \ref{Fig:limit_antipodal} for an illustration.
Note that these two conditions imply that if $\lim_{x\rightarrow\infty}f\neq 0$ and if the graph of $f$ is never tangent to the $x$-axis, the zero set consists of an odd number of components.
For any subset $A$ of a vector bundle $\xi=(E,B,\pi)$, denote by $Z(A)$ the set of base points on whose fiber $A$ contains $0$ or $A$ is unbounded.
In particular, for any section $s$, $Z(s)$ denotes the set of zeroes of the section (as a section is a single point on every fiber, and thus never unbounded).

\begin{figure}
\centering
\includegraphics[scale=0.7]{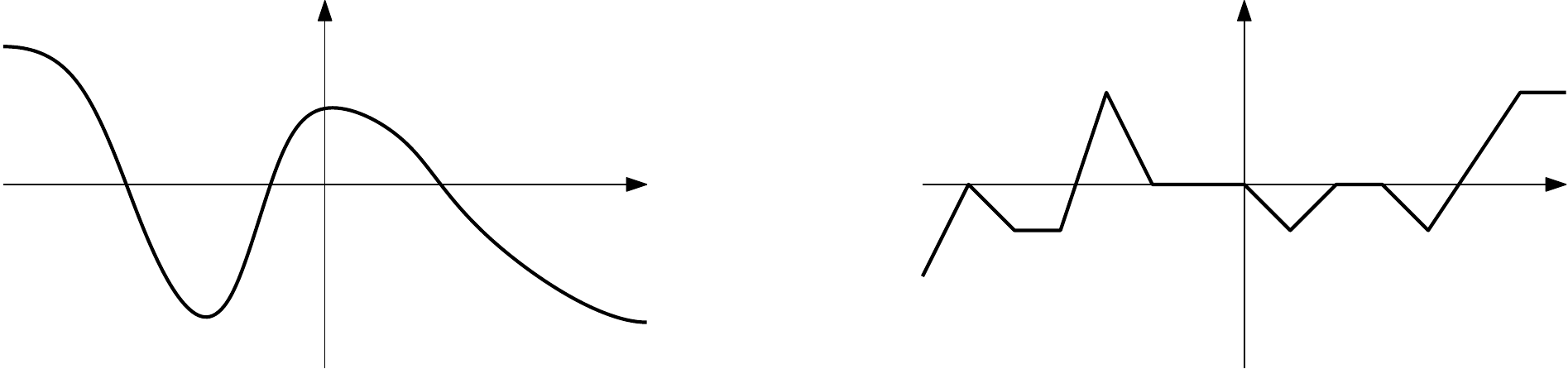}
\caption{Two graphs of limit antipodal functions.}
\label{Fig:limit_antipodal}
\end{figure}

Consider again $B:=G_{1}(\mathbb{R}^{d})^n$, i.e., all $n$-tuples of lines in $\mathbb{R}^d$ through the origin and the vector bundles $\xi_i$.
Note that $\xi_i$ has a natural orientable cover $\xi'_i=(E', B', \pi')$ where all the lines are oriented.
Denote by $p$ the covering map from $\xi'_i$ to $\xi_i$.

Assume now that we are given a continuous function $f: E'\rightarrow\mathbb{R}$ with the following properties:
\begin{enumerate}
\item[(a)] for every point $b'\in B$, the restriction of $f$ to the fiber $\pi^{-1}(b')$, denoted by $f_{b'}$, is a limit antipodal function;
\item[(b)] for any point $b\in B$ and any two lifts $b'_1, b'_2\in p^{-1}(b)$ we have either $f_{b'_1}(x)=f_{b'_2}(x)$ or $f_{b'_1}(x)=-f_{b'_2}(x)$ or $f_{b'_1}(x)=f_{b'_2}(-x)$ or $f_{b'_1}(x)=-f_{b'_2}(-x)$.
\end{enumerate}
Let $V'_f:=\{e\in E'|f(e)=0\}$ be the zero set of $f$.
Note that the second condition ensures that $V'_f$ is the lift of a set $V_f\subseteq E$.
We call $V_f$ a \emph{quasi-section} in $\xi_i$.
Further note that $Z(V_f)$ consists of the base points where $f_b(0)=0$ or $\lim_{x\rightarrow\infty}f_b=0$.

\begin{lemma}
\label{lem:quasisection}
Let $V_f$ be a quasi-section in $\xi_i$.
Then there is a section $s$ such that $Z(s)\subset Z(V_f)$. In particular, if $Z(V_f)=\emptyset$, then $\xi_i$ allows a nowhere zero section.
\end{lemma}

Before proving this lemma, we show how to apply it to prove Theorem \ref{Thm:pizza_origin}.

\begin{proof}[Proof of Theorem \ref{Thm:pizza_origin}]
Define $h_i(b)$ and $g_i(b,v)$ as in the proof of Theorem \ref{Thm:product}.

Consider now the mass $\mu_1$.
For each $b\in B$, choose some orientations of $h_2(b),\ldots, h_n(b)$ and an orientation of $\ell_1(b)$ arbitrarily.
Then for each $v\in\ell_1(b)$, we have well-defined regions $R^+(b,v)$ and $R^-(b,v)$.
In particular, taking $\mu_1(R^+(b,v))-\mu_1(R^-(b,v))$ for all orientations defines a function $f_1: E'\rightarrow\mathbb{R}$ which satisfies condition (a) and (b) from above.
Let $V_1$ be the set of all pairs $(b,v)$ such that $(g_1(b,v), h_2(b), h_3(b),\ldots, h_n(b))$ bisects $\mu_1$.
As this is exactly the set of pairs $(b,v)$ for which $f_1(b,v)=0$, it follows that $V_1$ is a quasi-section.

Let now $s_1^1$ be a section in $\xi_1$ with $Z(s_1)\subset Z(V_1)$, the existence of which we get from Lemma \ref{lem:quasisection}.
Analogously, use $\mu_i$ to define $V_i$ and $s_i^1$ for all $i\in\{2,\ldots,n\}$.
Then $s^1:=(s_1^1,\ldots,s_n^1)$ is a section in $(\gamma_{1}^{d})^n$.
Similarly, for $k\in{2,\ldots,d-1}$, using the masses $\mu_{(k-1)n+1},\ldots,\mu_{kn}$ define a section $s^k$ in $(\gamma_{1}^{d})^n$.

We have thus defined $d-1$ sections in $(\gamma_{1}^{d})^n$.
Hence, by applying Lemma \ref{lem:product}, we get that there is a point $b_0$ in $B$ such that $s^1(b_0),\ldots,s^{d-1}(b_0)=0$.
We claim that $H:=(h_1(b_0),\ldots,h_n(b_0))$ almost simultaneously bisects $\mu_1,\ldots,\mu_{(d-1)n}$: without loss of generality, consider the mass $\mu_1$.
As $s_1^1(b_0)=0$, we know by the definition of $s_1^1$ that $(b_0,0)$ is in $Z(V_1)$.
By the definition of $Z(V_1)$ this means that $V_1\cap\pi^{-1}(b_0)$ (1) contains $(b_0,0)$ or (2) is unbounded.

In case (1), we get that $(g_1(b_0,0), h_2(b_0),\ldots, h_n(b_0))$ bisects $\mu_1$.
But since $g_i(b_0,0)=h_i(b_0)$, this set is exactly $H$.
In case (2), we notice that $V_1$ is unbounded on $\pi^{-1}(b_0)$ if and only if $\lim_{x\rightarrow\infty}f_{1,b_0}=0$.
But this means that $(h_2(b_0),\ldots, h_n(b_0))$ bisects $\mu_1$.
Thus, $H$ indeed almost bisects $\mu_1$.
\end{proof}

From Theorem \ref{Thm:pizza_origin} we also deduce the main result of this section:

\pizza*

\begin{proof}
Map $\mathbb{R}^d$ to the hyperplane $p: x_{d+1}=1$ in $\mathbb{R}^{d+1}$.
This induces an embedding of the masses $\mu_1,\ldots,\mu_{dn}$.
By defining $\mu'_i(S)=\mu(S\cap p)$ for every full-dimensional open subset of $\mathbb{R}^{d+1}$, we get $dn$ mass distributions $\mu'_1,\ldots,\mu'_{dn}$ in $\mathbb{R}^{d+1}$.
By Theorem \ref{Thm:pizza_origin}, there are $n$ hyperplanes $\ell'_1,\ldots\ell'_n$ of dimension $d$ through the origin that almost simultaneously bisect $\mu'_1,\ldots,\mu'_{dn}$.
Define $\ell_i:=\ell'_i\cap p$.
Note that each $\ell_i$ is a hyperplane of dimension $d-1$.
By the definition of $\mu'_i$, the hyperplanes $\ell_1,\ldots\ell_n$ then almost simultaneously bisect $\mu_1,\ldots,\mu_{dn}$.
\end{proof}

It remains to prove Lemma \ref{lem:quasisection}.

\begin{proof}[Proof of Lemma \ref{lem:quasisection}]
Consider again the bundle $\xi'_i=(E', B', \pi')$, which is a cover of $\xi_i$.
The set $Z(V'_f)$ partitions $B'\setminus Z(V'_f)$ into connected components.
Consider two lifts $b'_1,b'_2$ of a point $b\in B$ with the property that the marked line $\ell_i$ is oriented differently in $b'_1$ than in $b'_2$.
We will call a pair of such lifts \emph{antipodal}.
We claim that if $b'_1,b'_2\not\in Z(V'_f)$ then $b'_1$ and $b'_2$ are not in the same connected component.
If this is true, then we can assign $1$ or $-1$ to each connected component in such a way that for any antipodal pair $b'_1, b'_2$, whenever we assign $1$ to the connected component containing $b'_1$ we assign $-1$ to the connected component containing $b'_2$.
We the define $s'$ as follows: for every $b'$, let $d(b')$ be the distance to the boundary of its connected component (note that there are several ways to define distance measures on $B'$, any of them is sufficient for our purposes).
Place a point at distance $d(b')$ from the origin on the positive side of $\ell_i$ if the connected component containing $b'$ was assigned a $1$, and on the negative side otherwise.
This gives a section on $\xi'$.
Further, for any two antipodal lifts $b'_1, b'_2$, we have $s(b'_1)=-s(b'_2)$.
Also, for any two lifts $b'_3, b'_4$, that are not antipodal, that is, $\ell_i$ is oriented the same way for both of them, we have $s(b'_3)=s(b'_4)$.
Thus, $s'$ projects to a section $s$ in $\xi$ with the property that $s(b)=0$ only if $b\in Z(V_f)$, which is want we want to prove.

Hence, we only need to show that a pair $b'_1, b'_2$ of antipodal lifts is not in the same connected component.
To this end, we will show that every path in $B'$ from $b'_1$ to $b'_2$ crosses $Z(V_f)$.
Let $\gamma$ be such a path.
Then $\gamma$ induces a continuous family of limit antipodal functions $f_t$, $t\in[0,1]$, with $f_0=f_{b'_1}$ and $f_1=f_{b'_2}$.
Further, as $b'_1$ and $b'_2$ are antipodal, we have $f_0(x)=\pm f_1(-x)$.
If for any $t$ we have $\lim_{x\rightarrow\infty}f_t=0$ we are done, so assume otherwise.
Then, it is not possible that $f_0(x)=f_1(-x)$, as in this case $\lim_{x\rightarrow\infty}f_0=-\lim_{x\rightarrow\infty}f_1$ , so by continuity, there must be a $t$ with $\lim_{x\rightarrow\infty}f_t=0$ .
Thus, assume that we have $f_0(x)=-f_1(x)$.

The set of zeroes of the $f_t$ defines a subset of $\mathbb{R}\times [0,1]$, which we denote by $W$.
See Figure \ref{Fig:path_family} for an illustration.
In general $W$ is not connected, but has finitely many connected components, as by the second condition for limit antipodality each $f_t$ has finitely many connected components of zeroes.
We say that a connected component $W_i$ of $W$ has \emph{full support} if for every $t\in [0,1|$, $f_t$ has a zero in $W_i$.
It can be deduced from the limit antipodality of the $f_t$'s that $W$ has an odd number of connected components with full support, denoted by $W_1,\ldots,W_{2k+1}$.
Consider the median component $W_{k+1}$.
Without loss of generality, $W_{k+1}$ is a path in $\mathbb{R}\times [0,1]$ from $(x,0)$ to $(-x,1)$.
By a simple continuity argument, we see that $W_{k+1}$ must cross the line $(0,t), t\in[0,1]$.
At this crossing, we are at a base point $b'\in Z(V_f)$, which concludes the proof.
\end{proof}

\begin{figure}
\centering
\includegraphics[scale=0.65]{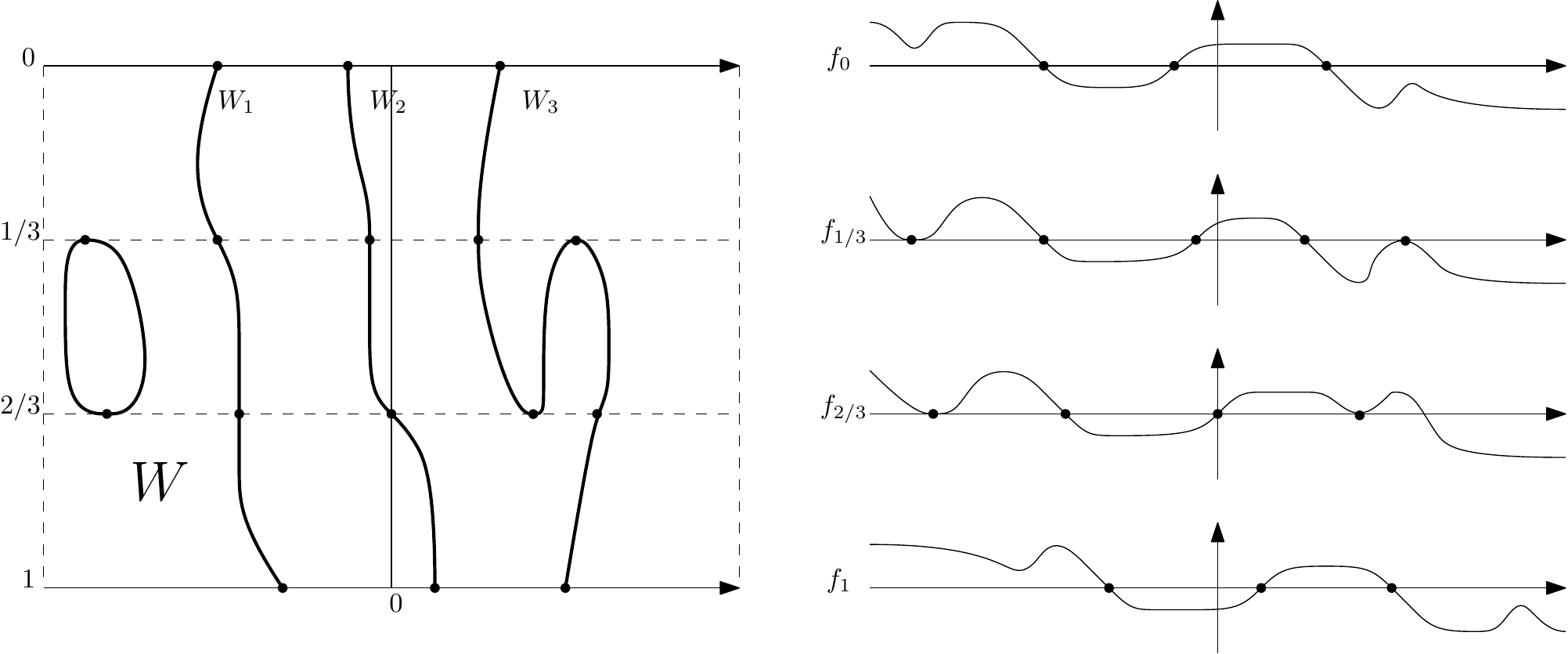}
\caption{The set $W$ for a family of limit antipodal functions between two antipodal lifts.}
\label{Fig:path_family}
\end{figure}

In order to prove Conjecture \ref{Conj:langerman}, we would like to choose $Z(V_f)$ as the set of base points where $f_b(0)=0$.
Let us briefly give an example where our arguments fail for this definition.
Consider $\mu$ as the area of a unit disk in $\mathbb{R}^2$.
If we want to simultaneously bisect $\mu$ with two lines $\ell_1$, $\ell_2$ through the origin, these lines need to be perpendicular.
Further, any single line through the origin bisects $\mu$ into two equal parts.
Imagine now the line $\ell_1$ to be fixed, and consider the limit antipodal function $f_b$ defined by sweeping $\ell_2$ along an oriented line perpendicular to $\ell_1$.
Without loss of generality, this function can be written as

\[
f_b(x) =
\begin{cases}
0 &\quad x\in [-\infty, -1] \\
1+x &\quad x\in [-1, 0] \\
1-x &\quad x\in [0, 1] \\
0 &\quad x\in [1, \infty]. \\
\end{cases}
\]

Note that this holds whenever $\ell_1$ and the sweep line for $\ell_2$ are perpendicular, so in particular, continuously rotating the arrangement by $180^\circ$ induces a path between two antipodal lifts in the cover.
Further, along this path we never had $f_b(0)=0$, so the two antipodal lifts would be in the same connected component, which would break the proof of Lemma \ref{lem:quasisection} under this definition of $Z(V_f)$.
Thus, conjecture \ref{Conj:langerman} remains open for now.

%

\bibliographystyle{plainurl} %
\bibliography{refs}

\end{document}